\documentclass[11pt]{article}

\usepackage[margin=1in]{geometry}
\usepackage[T1]{fontenc}
\usepackage[utf8]{inputenc}
\usepackage{lmodern}
\usepackage{microtype}
\usepackage{setspace}
\usepackage[nocompress]{cite}
\usepackage{amsmath,amssymb,amsthm}
\usepackage{mathtools}

\usepackage{authblk}
\usepackage{algorithm}
\usepackage{algorithmic}
\usepackage{graphicx}
\usepackage{placeins}

\usepackage{tikz}
\usetikzlibrary{matrix,arrows,positioning}

\usepackage[hidelinks]{hyperref}
\usepackage[nameinlink,capitalize,noabbrev]{cleveref}
\usepackage{xspace}

\newtheorem{theorem}{Theorem}
\newtheorem{lemma}[theorem]{Lemma}
\newtheorem{proposition}[theorem]{Proposition}
\newtheorem{corollary}[theorem]{Corollary}
\theoremstyle{definition}
\newtheorem{definition}[theorem]{Definition}
\newtheorem{example}[theorem]{Example}
\theoremstyle{remark}
\newtheorem{remark}[theorem]{Remark}

\newcommand{\Turnpike}{Turnpike\xspace}
\newcommand{\Beltway}{Beltway\xspace}

\newcommand{\MILP}{MILP\xspace}
\newcommand{\ILP}{ILP\xspace}
\newcommand{\LP}{LP\xspace}

\newcommand{\N}{\mathbb{N}}
\newcommand{\Z}{\mathbb{Z}}

\newcommand{\R}{\mathbb{R}}

\newcommand{\Binary}{\{0,1\}}
\newcommand{\indices}[1]{\left[#1\right]}

\newcommand{\Y}{Y}
\newcommand{\Ytilde}{\widetilde{Y}}
\newcommand{\Yhat}{\widehat{Y}}

\newcommand{\Phat}{\widehat{P}}

\newcommand{\rR}{(r,R)}

\newcommand{\Nx}{\indices{n}}
\newcommand{\NxxN}{\mathcal{I}_n}     % monotone intervals (i,j) with i<j
\newcommand{\NxxNxxN}{\mathcal{T}_n}  % monotone refinements (i,j,k) with i<j<k

\newcommand{\ruler}{\rho}

\theoremstyle{definition}
\newtheorem{system}[theorem]{System}

\title{\textbf{Turnpike with Uncertain Measurements:\\
Triangle-Equality ILP with a Deterministic Recovery Guarantee}}

\author[1]{C.~S.~Elder}
\author[1]{Guillaume~Mar\c{c}ais}
\author[1]{Carl~Kingsford}

\affil[1]{Ray and Stephanie Lane Computational Biology Department, Carnegie Mellon University\\
5000 Forbes Ave., Pittsburgh, PA 15213, USA\vspace{-1ex}}
\date{\today}
\begin{document}
\maketitle

\begin{abstract}
We study \Turnpike with uncertain measurements: reconstructing a one-dimensional point set from an unlabeled multiset of pairwise distances under bounded noise and rounding.
We give a combinatorial characterization of realizability via a multi-matching that labels interval indices by distinct distance values while satisfying all triangle equalities.
This yields an \ILP based on the triangle equality whose constraint structure depends only on the two-partition set $\mathcal{P}_y=\{(r,s,t): y_r+y_s=y_t\}$ (with $\mu_r\ge 2$ when $r=s$) and a natural \LP relaxation with $\{0,1\}$-coefficient constraints.
Integral solutions certify realizability and output an explicit assignment matrix, enabling a modular \emph{assignment-first, regression-second} pipeline for downstream coordinate estimation.
Under bounded noise followed by rounding, we prove a deterministic separation condition under which $\mathcal{P}_y$ is recovered exactly, so the \ILP/\LP receives the same combinatorial input as in the noiseless case.
Experiments illustrate integrality behavior and degradation outside the provable regime.
\end{abstract}


\section{Introduction}
\label{sec:introduction}

The \emph{\Turnpike problem} is a classical inverse problem: reconstruct a one-dimensional point set from an unlabeled multiset of pairwise distances.
The origin story imagines a highway: we are given the multiset of all pairwise distances between exits, but the exit mile markers themselves have been lost. 
Can we recover a configuration of exits that is consistent with the
observed distances?

Formally, let $x \in \R^n$ have strictly increasing coordinates. 
Then the \emph{difference multiset} associated with $x$ is \(
    \Delta x \coloneqq \{x_j-x_i : 1\le i<j\le n\}
\), which is a multiset containing $m \coloneqq \binom{n}{2}$ positive reals.
An instance of \Turnpike is a multiset $Y$ of $m$ positive reals, and the decision problem asks whether there exists such an $x$ with $\Delta x=Y$.

A multiset $\Y$ is represented by a pair $(y, \mu)$ such that $y$ lists the values contained in $\Y$ in strictly decreasing order \(
    y_1 > \dots > y_{m'} >0
\) where the multiplicity vector $\mu \in \N^{m'}_{++}$ records the number of occurrences, i.e., that a value $y_k$ is present $\mu_k \ge 1$ times.
We say $\Y$ is \emph{realizable} when $\Y=\Delta x$ for some $x \in \R^n$.

\Turnpike is equivalent to the Partial Digest problem from restriction site mapping in computational biology~\cite{skiena_PDP_1994,alizadeh_digestion_protocol_1995} and is closely related to distance-geometry reconstruction tasks in crystallography, tandem mass spectrometry, and molecular error-correcting codes~\cite{Huang_ReconstructingPointSetsDistributions_2021,fomin_cyclopeptide_circular_sums_2015,gabrys_polymer_ECCs_turnpike_2020}.
In restriction mapping, for example, a single-enzyme partial digestion produces fragments whose lengths coincide with all pairwise distances between restriction sites; the multiset $Y$ records fragment lengths, and a solution $x$ specifies restriction site locations along the strand (up to translation and reflection).

Exact recovery from error-free distances is neither known to be NP-complete nor to admit a polynomial-time algorithm~\cite{chen_pairwise_reconstruction_2009}.
A backtracking algorithm with worst-case exponential runtime but polynomial expected runtime under absolutely-continuous input models was given in~\cite{skiena_distances_1990,skiena_PDP_1994}.
Many additional approaches have been proposed, including polynomial factorization methods for exact instances, pruning strategies, and hardness and counterexample constructions for related variants~\cite{dakic_turnpike_2000,blazewicz_SPDP_hardness_2009,chen_pairwise_reconstruction_2009,jaganathan_pairwise_distances_2013,nadimi_fast_PDP_2011,wang_fast_SPDP_2023}.

\paragraph*{Noisy \Turnpike and measurement uncertainty}
Experimental pipelines that naturally produce \Turnpike-like instances (e.g., partial digestion and tandem mass
spectrometry) inevitably include measurement uncertainty, as well as missing or duplicated observations.
A standard formulation is \emph{bounded-error \Turnpike}: the input is an observed multiset $\Ytilde$ together
with an error radius $r\ge 0$, and the question is whether there exist points $x$ and a matching between the distances $\Delta x$ and the observations $\Ytilde$ such that each matched pair differs by at most~$r$.
Allowing bounded additive or multiplicative errors makes Partial Digest (and hence \Turnpike) strongly NP-hard~\cite{cieliebak_UPDP_hardness_2005}, with related hardness results holding for many other noisy
variants~\cite{Huang_ReconstructingPointSetsDistributions_2021,elder_beltway_turnpike_2024}.
A closely related circular analogue is the \Beltway problem; reductions connect noisy \Beltway and noisy
\Turnpike~\cite{elder_beltway_turnpike_2024}.
In the presence of uncertainty, backtracking approaches routinely explore exponentially many states, even on
inputs sampled from non-adversarial distributions~\cite{skiena_PDP_1994,cieliebak_UPDP_hardness_2005,Huang_ReconstructingPointSetsDistributions_2021,elder_beltway_turnpike_2024}.

These noisy variants are strongly NP-hard, so we do not aim for an exact, polynomial-time algorithm.
Instead, we focus on formulations whose integral solutions certify realizability and whose fractional solutions still retain structured combinatorial information.

\subsection{This work: a triangle-based assignment viewpoint}
\label{subsec:this-work}

This work develops a complementary viewpoint that treats unlabeled distances as a \emph{combinatorial labeling problem} first and a regression problem second.
The key object is a \emph{ruler}: an indexed pairwise-difference table in which entry $(i,j)$ stores $x_j-x_i$.
Unlike the multiset $\Delta x$, a ruler retains interval indices, and the identity $x_k-x_i=(x_j-x_i)+(x_k-x_j)$ becomes a family of linear consistency constraints (the \emph{triangle equalities}).
Realizability of $\Y$ can therefore be rephrased as the existence of an assignment of distance values to interval indices that satisfies triangle equalities on every triple $i<j<k$.

This lens leads to a natural \MILP formulation for exact \Turnpike, which we reformulate into an \ILP using only combinatorial information.
A distinctive feature of this formulation is that it outputs an \emph{assignment matrix} $\widehat P$ (a multi-matching between intervals and distinct distance values). 
A coordinate estimate $\widehat x$ can then be computed afterward by solving an independent regression problem with any preferred numeric precision.
This makes permutation metrics over $\widehat P$ the natural choice for primary evaluation, with coordinate-wise metrics treated as secondary.

In summary, we build a triangle-equality \ILP for exact \Turnpike, show that this formulation extends to the uncertain setting under a deterministic two-partition recovery condition, and empirically validate that performance tracks this regime and degrades rapidly outside it.

\paragraph*{Main results and algorithmic viewpoint}
\begin{itemize}
  \item \textbf{Triangle feasibility as realizability.}
  We characterize realizability of $\Y$ as the existence of a multi-matching between interval indices and distinct distance values whose induced ruler satisfies all triangle equalities (Lemma~\ref{lem:realizable-matching}). 
  This yields a simple \MILP formulation (Proposition~\ref{prop:MILP-turnpike}).

  \item \textbf{A triangle-based \ILP/\LP parameterized by two-partitions.}
  We eliminate explicit coordinate variables by expressing triangle consistency solely through additive
  relations among the distinct distance values. The resulting \ILP (Proposition~\ref{prop:triangle-ilp}) and \LP relaxation (System~\ref{sys:triangle-lp}) are parameterized by the two-partition set
  $\mathcal{P}_y=\{(r,s,t): y_r+y_s=y_t\}$ (with $\mu_r\ge 2$ when $r=s$), which can be enumerated in $O(m'^2)$ time by a two-pointer scan (Algorithm~\ref{alg:two-pointer}).
  The formulation uses $O(n^2m') = O(n^4)$ assignment variables (since $m' = n(n-1)/2$ in the worst case) and---when using the basis reduction of Lemma~\ref{lem:triangle-basis}---$O(n^2|\mathcal{P}_y|)$ triangle variables; since $|\mathcal{P}_y|=O(m'^2)$, this is $O(n^2m'^2) = O(n^6)$.

  \item \textbf{Deterministic robustness under bounded noise and rounding.}
  For bounded-error instances followed by rounding with radius $R$, we prove a separation condition ensuring that the exact two-partition structure is preserved. 
  When $\mathrm{gap}_\star>6(r+R)$, where $\mathrm{gap}_\star$ lower-bounds $|y_r+y_s-y_t|$ over all triples with $y_r+y_s\ne y_t$ (formal definition in Section~\ref{subsec:two-partition-recovery}), the rounded data recover the same $\mathcal{P}_y$ as the noiseless instance (Theorem~\ref{thm:two-partition-recovery}), i.e., the
  triangle-equality formulation receives the same input as that of the ground truth.

  \item \textbf{Empirical evidence and an assignment-first evaluation philosophy.}
  We include experiments on synthetic instances and partial-digest benchmarks illustrating integrality behavior across synthetic distributions and noise/rounding regimes, and the resulting assignment matrices. 
  The experiments emphasize metrics on assignments (interval labels and permutation distances) rather than on downstream coordinate regression.
  In particular, the $(r,R)$ phase diagram in Figure~\ref{fig:noisy-phase} exhibits rapid degradation in two-partition recovery once the separation regime of Theorem~\ref{thm:two-partition-recovery} is violated.
\end{itemize}

Limitations: our deterministic guarantee assumes bounded noise followed by rounding and does not cover missing or duplicated distances.
Although the formulations use polynomially many constraints and variables in $n$ and $m'$, they remain impractically large beyond approximately $100$ points; we restrict our experiments accordingly.
Finally, extending the approach to broader experimental noise models is left for future work.

\section{Preliminaries}
\label{sec:preliminaries}

This section fixes notation and introduces the objects used throughout: distance multisets, interval indices, rulers, and the multi-matching viewpoint that converts realizability into an assignment-feasibility question.
We conclude by formalizing the $(r,R)$ noise-and-rounding model used later for deterministic recovery guarantees.

\paragraph*{Quick reference (objects and notation).}
We use $\indices{n}=\{1,2,\dots,n\}$ and represent the input distances as a multiset $\Y=(y,\mu)$ with distinct values $y_1>\cdots>y_{m'}>0$ (total multiplicity $m=\binom{n}{2}$).
Intervals are $\NxxN=\{(i,j)\in\indices{n}^2 : i<j\}$; a multi-matching labels each interval by an index $r\in[m']$, yielding an assignment matrix $(P^r_{ij})$.
Triangle equalities are encoded by the two-partition set $\mathcal{P}_y=\{(r,s,t): y_r+y_s=y_t\}$ (with $\mu_r\ge 2$ when $r=s$), which parameterizes our triangle \ILP/\LP.
In the $(r,R)$ model, $r$ bounds measurement error and $R$ is the rounding radius.

\subsection{Notation and multisets}
\label{subsec:notation}

We write $\R$, $\Z$, and $\N$ for the reals, integers, and naturals, respectively; subscripts indicate
standard restrictions (e.g., $\R_+$ and $\R_{++}$ for nonnegative and positive reals). For an integer $n\ge 1$
we use $\indices{n}=\{1,2,\dots,n\}$.

A numeric multiset is represented as an ordered pair $\Y=(y,\mu)$ where $y\in\R^{m'}$ lists the distinct
values in strictly decreasing order $y_1>\cdots>y_{m'}>0$, and $\mu\in\N^{m'}_{++}$ records the multiplicity of
each value. The total cardinality (counting multiplicity) is $m=\sum_{r=1}^{m'}\mu_r$.

As in Section~\ref{sec:introduction}, for $x\in\R^n$ with strictly increasing entries we write
$\Delta x=\{x_j-x_i: 1\le i<j\le n\}$ for the multiset of pairwise distances. When $\Delta x=(y,\mu)$, the
total multiplicity satisfies $m=\binom{n}{2}$.

\begin{example}
For $x=(0,2,5,9)$, the distances are $\{2,3,4,5,7,9\}$, so $m'=m=6$ and $\mu=\mathbf{1}$.
For $x=(0,2,4,6)$, the distances are $\{2,2,2,4,4,6\}$, so $y=(6,4,2)$ and $\mu=(1,2,3)$.
\end{example}

\subsection{Intervals, refinements, and rulers}
\label{subsec:rulers-prelim}

We index unknown distances by their endpoints.
A \emph{monotone interval} is an ordered pair in \(
    \NxxN \coloneqq \{(i,j)\in\Nx\times\Nx : i<j\}
\); a \emph{monotone refinement} is an ordered triple in \(
    \NxxNxxN \coloneqq \{(i,j,k)\in\Nx^3 : i<j<k\}.
\)
When unambiguous, we shorten $(i, j)$ to $ij$ and $(i,j,k)$ to $ijk$.

\begin{definition}
  \label{def:ruler}
  A \emph{ruler} on $\Nx$ is a matrix $\ruler\in\R^{\Nx\times\Nx}$ for which there exists $x\in\R^n$ such that
  $\ruler_{ij}=x_j-x_i$ for all $i,j\in\Nx$. We call $x$ a \emph{realization} of $\ruler$.
  The ruler is \emph{monotone} if $\ruler_{ij}>0$ for all $ij\in\NxxN$.
\end{definition}

Every ruler satisfies skew-symmetry ($\ruler_{ij}=-\ruler_{ji}$) and the triangle equalities
$\ruler_{ik}=\ruler_{ij}+\ruler_{jk}$.
The converse also holds.

\begin{lemma}[Triangle equalities characterize rulers]
  \label{lem:ruler-triangle}
  A matrix $\ruler\in\R^{\Nx\times\Nx}$ is a ruler if and only if it satisfies
  $\ruler_{ik}=\ruler_{ij}+\ruler_{jk}$ for all $i,j,k\in\Nx$.
\end{lemma}

\begin{proof}
If $\ruler_{ij}=x_j-x_i$ for some $x$, then
$\ruler_{ik}=(x_k-x_i)=(x_j-x_i)+(x_k-x_j)=\ruler_{ij}+\ruler_{jk}$.

Conversely, suppose $\ruler$ satisfies the triangle equalities for all triples.
Setting $k=i$ gives $\ruler_{ii}=\ruler_{ij}+\ruler_{ji}$, so $\ruler_{ii}=0$ and $\ruler_{ji}=-\ruler_{ij}$.
Define $x_i \coloneqq \ruler_{1i}$. Then for any $i,j$,
$\ruler_{ij}=\ruler_{i1}+\ruler_{1j}=-(\ruler_{1i})+\ruler_{1j}=x_j-x_i$.
\end{proof}

\begin{remark}
Lemma~\ref{lem:ruler-triangle} shows that triangle equalities alone force the expected normalization
$\ruler_{ii}=0$ and skew-symmetry. In particular, a monotone ruler determines a strictly increasing point set
up to translation.
\end{remark}

Because pairwise differences are translation invariant, realizability is unaffected by global shifts of $x$:
replacing $x$ by $x+c\mathbf{1}$ leaves $\Delta x$ unchanged.

\subsection{Multi-matchings and realizability}
\label{subsec:matching}

We now express realizability of $\Y=(y,\mu)$ as the existence of a multi-matching between intervals and
distinct distance values.

\begin{definition}
  \label{def:multimatching}
  Let $\mathcal{X}$ be a finite set, and let $\Y=(y,\mu)$ be a multiset with distinct values
  $y_1>\cdots>y_{m'}>0$.
  A \emph{multi-matching} from $\mathcal{X}$ to $\Y$ is a map $P:\mathcal{X}\to[m']$ such that
  $|P^{-1}(r)|=\mu_r$ for every $r\in[m']$.
\end{definition}

When $\mathcal{X}=\NxxN$, the map $P$ assigns each interval $ij$ to one of the distinct distance values. We
often identify $P$ with its one-hot encoding: binary variables $P^r_{ij}$ indicate whether interval $ij$ is
assigned value $y_r$.

Given a multi-matching $P:\NxxN\to[m']$, define the induced \emph{oriented difference matrix}
$\ruler^{(P)}\in\R^{\Nx\times\Nx}$ by
\[
  \ruler^{(P)}_{ij}=
  \begin{cases}
    y_{P_{ij}}, & i<j,\\
    0,          & i=j,\\
    -y_{P_{ji}},& i>j.
  \end{cases}
\]

\begin{lemma}
  \label{lem:realizable-matching}
  A multiset $\Y=(y,\mu)$ is realizable if and only if there exists a multi-matching $P:\NxxN\to[m']$ with
  multiplicity~$\mu$ such that $\ruler^{(P)}$ is a monotone ruler.
\end{lemma}

\begin{proof}
If $\Y$ is realizable, choose $x\in\R^n$ with strictly increasing entries and $\Delta x=\Y$. For each
interval $ij\in\NxxN$, pick the unique index $r$ with $x_j-x_i=y_r$ and set $P_{ij} \coloneqq r$. The multiplicities
in $\Delta x$ ensure $|P^{-1}(r)|=\mu_r$. By construction, $\ruler^{(P)}_{ij}=x_j-x_i$, so $\ruler^{(P)}$ is a
monotone ruler.

Conversely, if $P$ is a multi-matching and $\ruler^{(P)}$ is a monotone ruler, then by
Lemma~\ref{lem:ruler-triangle} there exists $x\in\R^n$ with $\ruler^{(P)}_{ij}=x_j-x_i$. Monotonicity implies
$x$ is strictly increasing up to translation. For $ij\in\NxxN$ we have $x_j-x_i=\ruler^{(P)}_{ij}=y_{P_{ij}}$,
and the multi-matching multiplicities imply $\Delta x=\Y$.
\end{proof}

\begin{example}[One triangle constraint]
For $n=4$, the refinement $(1,2,4)$ enforces $(x_2-x_1)+(x_4-x_2)=x_4-x_1$.
In the assignment view, this means the labels on intervals $12$, $24$, and $14$ must form a two-partition
$(r,s,t)\in\mathcal{P}_y$ with $y_r+y_s=y_t$.
\end{example}

\subsection{An \texorpdfstring{\MILP}{MILP} formulation}
\label{subsec:milp-turnpike}

Lemma~\ref{lem:realizable-matching} leads to a natural mixed-integer formulation that couples a discrete
assignment with continuous coordinates.

\begin{proposition}
  \label{prop:MILP-turnpike}
  Let $\Y=(y,\mu)$ have total multiplicity $m=\binom{n}{2}$. The following \MILP\ is feasible if and only if
  $\Y$ is realizable:
  \begin{align*}
    &x\in\R^n,\qquad P^r_{ij}\in\Binary &&\forall\, ij\in\NxxN,\ r\in[m'],\\
    &\sum_{r=1}^{m'} P^r_{ij}=1         &&\forall\, ij\in\NxxN,\\
    &\sum_{ij\in\NxxN} P^r_{ij}=\mu_r   &&\forall\, r\in[m'],\\
    &x_j-x_i=\sum_{r=1}^{m'} y_r P^r_{ij} &&\forall\, ij\in\NxxN.
  \end{align*}
  Optionally one may fix translation by adding $x_1=0$.
\end{proposition}

\begin{proof}
If the system is feasible, define $P_{ij}$ as the unique $r$ with $P^r_{ij}=1$. The last constraint gives
$x_j-x_i=y_{P_{ij}}$, so $\Delta x=\Y$ by the multiplicity constraints.

Conversely, if $\Y$ is realizable, choose $x$ and a multi-matching $P$ as in Lemma~\ref{lem:realizable-matching}
and set $P^r_{ij}=\mathbf{1}[P_{ij}=r]$. Then all constraints hold.
\end{proof}

Proposition~\ref{prop:MILP-turnpike} correctly captures realizability, but its LP relaxation depends on the
numerical encoding of $y$ through the constraints $x_j-x_i=\sum_r y_r P^r_{ij}$. In
Section~\ref{sec:triangle-lp} we derive a triangle-based formulation whose constraint structure depends
instead on additive relations among the distinct distance values.

\subsection{Noisy variants and the \texorpdfstring{$\rR$}{(r,R)} model}
\label{sec:noisy-variants}

We now formalize the bounded-error model and the rounding step used in our robustness guarantees.

In bounded-error noisy \Turnpike, the input consists of an observed multiset $\Ytilde=(\tilde y,\tilde\mu)$
together with an error radius $r\ge 0$. One asks whether there exist a point set $x\in\R^n$, a realizable
multiset $\Y=(y,\mu)$ with $\Delta x=\Y$, and a matching between the elements of $\Y$ and those of $\Ytilde$
such that each matched pair differs by at most~$r$.

To reason about additive relations among noisy distances, we also introduce a rounding radius $R\ge 0$. 
For each distinct observed value $\tilde y_\ell$, define its rounded value $\hat y_\ell \coloneqq \operatorname{round}_R(\tilde y_\ell)\in R\Z$, where $\operatorname{round}_R$ rounds to the nearest multiple of $R$ (ties may be broken arbitrarily). 
When $R=0$ we interpret $\operatorname{round}_0$ as the identity map.
After rounding, identical values are aggregated: we write $\Yhat=(\hat y,\hat\mu)$ for the resulting multiset of \emph{distinct} rounded values and their multiplicities.

The pair $\rR$ parameterizes our deterministic robustness statements: $r$ bounds measurement uncertainty while $R$ bounds discretization error introduced by rounding. When the distinct true distances are sufficiently separated relative to $r+R$, Section~\ref{subsec:two-partition-recovery} shows that the two-partition structure can be recovered exactly from $\Yhat$, enabling the triangle-based \ILP/\LP\ formulations to operate on the correct combinatorial input.

\section{A Triangle-Based \texorpdfstring{\ILP}{ILP} and \texorpdfstring{\LP}{LP} Relaxation}
\label{sec:triangle-lp}

The \MILP formulation in Proposition~\ref{prop:MILP-turnpike} enforces ruler consistency through explicit
numerical equalities of the form $x_j-x_i=\sum_r y_r P^r_{ij}$, which couple discrete assignment variables and
continuous coordinates.
In this section, we eliminate the coordinate variables entirely by enforcing ruler structure using only the additive relations among the distinct distance values.
This yields a triangle-based \ILP whose constraint structure depends only on indices and on the set of \emph{two-partitions} present in $y$, together with a natural \LP relaxation.

We first define two-partitions and $\mathcal{P}_y$, then present the triangle \ILP, its \LP relaxation, two structural reductions, and finally the robustness guarantee under noise and rounding.

\subsection{Two-partitions}
\label{subsec:two-partitions}

Let $\Y=(y,\mu)$ with distinct values $y_1 > \cdots > y_{m'} > 0$.
A \emph{two-partition} is an ordered triple \(
    (r,s,t) \in [m']^{3}
\) such that \(
  y_r + y_s = y_t,
\) with the additional requirement that $\mu_r \ge 2$ when $r = s$, ensuring the multiset contains two copies of $y_r$ when it is used twice.
We collect all such relations in \(
  \mathcal{P}_y \coloneqq \{(r,s,t)\in[m']^{3} : y_r+y_s=y_t\ \text{and}\ (r\neq s\ \text{or}\ \mu_r\ge 2)\}.
\)
This set can be enumerated in $O(m'^2)$ time by a two-pointer scan for each fixed target $t\in[m']$.

\begin{algorithm}
  \caption{Two-pointer enumeration of $\mathcal{P}_y(t)$ for a fixed target $t$}
  \label{alg:two-pointer}
  \footnotesize
  \begin{algorithmic}[1]
    \REQUIRE Distinct values $y_1>\cdots>y_{m'}>0$ with multiplicities $\mu$, and target $t\in[m']$.
    \ENSURE $\mathcal{P}_y(t)\coloneqq\{(r,s,t): y_r+y_s=y_t\ \text{and}\ (r\neq s\ \text{or}\ \mu_r\ge 2)\}$.
    \STATE $r\gets t+1$, $s\gets m'$, $P\gets\emptyset$
    \WHILE{$r<s$}
      \STATE $\sigma\gets y_r+y_s$
      \IF{$\sigma=y_t$}
        \STATE add $(r,s,t)$ and $(s,r,t)$ to $P$
        \STATE $r\gets r+1$, $s\gets s-1$
      \ELSIF{$\sigma>y_t$}
        \STATE $r\gets r+1$
      \ELSE
        \STATE $s\gets s-1$
      \ENDIF
    \ENDWHILE
    \IF{$r=s$ \textbf{and} $\mu_r\ge 2$ \textbf{and} $2y_r=y_t$}
      \STATE add $(r,r,t)$ to $P$
    \ENDIF
    \STATE \textbf{return} $P$
  \end{algorithmic}
\end{algorithm}

Running Algorithm~\ref{alg:two-pointer} for all $t\in[m']$ enumerates $\mathcal{P}_y$ in $O(m'^2)$ time, removing the extra $\log m'$ factor from a naive per-pair binary search.

\subsection{A triangle-based \ILP}
\label{subsec:triangle-ilp}

We now introduce an \ILP over assignment variables $P$ and triangle variables~$T$.
Intuitively, $P$ specifies which distance value labels each interval $ij$, while $T$ enforces that every
refinement $i<j<k$ respects the triangle equality through a compatible two-partition in~$\mathcal{P}_y$.

For each interval $ij\in\NxxN$ and each distance index $r\in[m']$, we use a binary variable $P^r_{ij}$.
For each refinement $ijk\in\NxxNxxN$ and each two-partition $(r,s,t)\in\mathcal{P}_y$, we introduce a binary
variable $T^{rst}_{ijk}$ indicating that the three intervals $(ij,jk,ik)$ are labeled by $(r,s,t)$.

\newpage
\begin{proposition}[Triangle-based \ILP]
  \label{prop:triangle-ilp}
  Let $\Y=(y,\mu)$ and $\mathcal{P}_y$ be the respective two-partition set. 
  The following \ILP is feasible if and only if $\Y$ is realizable:
  \begin{align*}
    &P^r_{ij}\in\Binary
      &&\forall\, ij\in\NxxN,\ r\in[m'],\\
    &T^{rst}_{ijk}\in\Binary
      &&\forall\, ijk\in\NxxNxxN,\ (r,s,t)\in\mathcal{P}_y,\\[2pt]
    &\sum_{r=1}^{m'} P^r_{ij}=1
      &&\forall\, ij\in\NxxN,\\
    &\sum_{ij\in\NxxN} P^r_{ij}=\mu_r
      &&\forall\, r\in[m'],\\[2pt]
    &\sum_{(r,s,t)\in\mathcal{P}_y} T^{rst}_{ijk}=1
      &&\forall\, ijk\in\NxxNxxN,\\[2pt]
    &P^r_{ij}=\sum_{\substack{(r',s,t)\in\mathcal{P}_y\\ r'=r}} T^{r'st}_{ijk}
      &&\forall\, ijk\in\NxxNxxN,\ \forall r\in[m'],\\
    &P^s_{jk}=\sum_{\substack{(r,s',t)\in\mathcal{P}_y\\ s'=s}} T^{rs't}_{ijk}
      &&\forall\, ijk\in\NxxNxxN,\ \forall s\in[m'],\\
    &P^t_{ik}=\sum_{\substack{(r,s,t')\in\mathcal{P}_y\\ t'=t}} T^{rst'}_{ijk}
      &&\forall\, ijk\in\NxxNxxN,\ \forall t\in[m'].
  \end{align*}
\end{proposition}

\begin{proof}
The first two constraint families enforce that $P$ is a multi-matching between intervals and distinct
distance values with multiplicities~$\mu$.
The remaining constraints enforce triangle consistency: for each refinement $i<j<k$, exactly one two-partition
$(r,s,t)\in\mathcal{P}_y$ is selected, and the induced labels on $(ij,jk,ik)$ agree with~$(r,s,t)$.

If $\Y$ is realizable, choose $x$ with $\Delta x=\Y$, and let $P$ be the induced interval labeling from
Lemma~\ref{lem:realizable-matching}. For each refinement $i<j<k$, the true distances satisfy
$(x_j-x_i)+(x_k-x_j)=x_k-x_i$, so if $x_j-x_i=y_r$, $x_k-x_j=y_s$, and $x_k-x_i=y_t$, then $(r,s,t)\in\mathcal{P}_y$.
Setting $T^{rst}_{ijk}=1$ for these triples (and $0$ otherwise) gives feasibility.

Conversely, suppose $(P,T)$ satisfies the \ILP. By construction, $P$ defines a multi-matching with the correct
multiplicities. The triangle constraints imply that for every refinement $i<j<k$, the three entries of the
induced difference matrix $\ruler^{(P)}$ satisfy the triangle equality. By Lemma~\ref{lem:ruler-triangle},
$\ruler^{(P)}$ is a ruler, hence corresponds to a point set $x$ with $\Delta x=\Y$.
\end{proof}

\subsection{Triangle-based \LP relaxation and interpretation}
\label{subsec:triangle-lp-system}

Relaxing integrality in Proposition~\ref{prop:triangle-ilp} yields a linear system.

\begin{system}[Triangle-based \LP relaxation]
  \label{sys:triangle-lp}
  The triangle-based \LP relaxation consists of variables
  $P^r_{ij}\in[0,1]$ for $ij\in\NxxN$ and $r\in[m']$, and variables
  $T^{rst}_{ijk}\in[0,1]$ for $ijk\in\NxxNxxN$ and $(r,s,t)\in\mathcal{P}_y$,
  together with the same linear equalities as in Proposition~\ref{prop:triangle-ilp}.
\end{system}

An integral solution $P$ to System~\ref{sys:triangle-lp} is feasible for the \ILP and therefore certifies realizability of~$\Y$.
A fractional solution may reflect either that $\Y$ is non-realizable or that the relaxation is inexact in this instance.

However, a feasible solution---fractional or otherwise---does induce a consistent \emph{ruler}, which is constructed as follows from the (possibly fractional) multi-matching $P$: set $\ruler^{(P)}_{ij}\coloneqq \sum_{r=1}^{m'} y_r P^r_{ij}$ for $ij\in\NxxN$ and then extend to $\ruler^{(P)}_{ii}=0$ and $\ruler^{(P)}_{ji}=-\ruler^{(P)}_{ij}$, which yields a ruler and hence a point set up to translation (Corollary~\ref{cor:lp-preserves-ruler}).
Thus, while the \LP relaxation is inexact in general, fractional solutions are structured, inducing valid ruler matrices that can be used in downstream regression tasks (e.g., as starting points for alternating schemes approximating \Beltway and \Turnpike~\cite{elder_beltway_turnpike_2024}).

\subsection{Structural consequences and formulation reductions}
\label{subsec:structural-reductions}

This subsection records several direct consequences of the triangle formulation. None introduces additional
coordinate variables.

\paragraph*{Fractional solutions induce consistent rulers.}
Let $(P,T)$ be feasible for System~\ref{sys:triangle-lp}. Define $\ruler^{(P)}\in\R^{\Nx\times\Nx}$ by setting,
for each $ij\in\NxxN$,
$\ruler^{(P)}_{ij}\coloneqq \sum_{r=1}^{m'} y_r P^r_{ij}$, with $\ruler^{(P)}_{ii}=0$ and
$\ruler^{(P)}_{ji}=-\ruler^{(P)}_{ij}$.

\begin{corollary}[Feasible \LP solutions induce a ruler]
\label{cor:lp-preserves-ruler}
If $(P,T)$ satisfies System~\ref{sys:triangle-lp}, then
$\ruler^{(P)}_{ik}=\ruler^{(P)}_{ij}+\ruler^{(P)}_{jk}$ for all $i<j<k$.
Consequently, there exists $x\in\R^n$, unique up to translation, such that
$\ruler^{(P)}_{ij}=x_j-x_i$ for all $i,j\in\Nx$.
\end{corollary}

\begin{proof}
Fix $i<j<k$.
Using the marginalization constraints in System~\ref{sys:triangle-lp}, we have
\begin{align*}
  \ruler^{(P)}_{ik}
  &= \sum_{t=1}^{m'} y_t P^t_{ik}
   = \sum_{(r,s,t)\in\mathcal{P}_y} y_t\, T^{rst}_{ijk} \\
  &= \sum_{(r,s,t)\in\mathcal{P}_y} (y_r+y_s)\, T^{rst}_{ijk}
   = \sum_{(r,s,t)\in\mathcal{P}_y} y_r\, T^{rst}_{ijk}
     + \sum_{(r,s,t)\in\mathcal{P}_y} y_s\, T^{rst}_{ijk} \\
  &= \sum_{r=1}^{m'} y_r P^r_{ij} + \sum_{s=1}^{m'} y_s P^s_{jk}
   = \ruler^{(P)}_{ij}+\ruler^{(P)}_{jk}.
\end{align*}
By construction, $\ruler^{(P)}_{ii}=0$ and $\ruler^{(P)}$ is skew-symmetric, so the same identity extends to
all triples $i,j,k\in\Nx$.
Lemma~\ref{lem:ruler-triangle} then implies $\ruler^{(P)}$ is a ruler.
\end{proof}

\paragraph*{A triangle basis.}
The full family of refinement-indexed triangle constraints is redundant.

\begin{lemma}[Triangle-basis reduction]
\label{lem:triangle-basis}
Let $\ruler\in\R^{\Nx\times\Nx}$ satisfy $\ruler_{ii}=0$ and $\ruler_{ji}=-\ruler_{ij}$.
The following are equivalent:
\begin{enumerate}
  \item $\ruler_{ik}=\ruler_{ij}+\ruler_{jk}$ for all $i<j<k$,
  \item $\ruler_{1k}=\ruler_{1j}+\ruler_{jk}$ for all $1<j<k$,
  \item there exists $x\in\R^n$ such that $\ruler_{ij}=x_j-x_i$ for all $i,j\in\Nx$.
\end{enumerate}
\end{lemma}

\begin{proof}
(1) $\Rightarrow$ (2) is immediate.
For (2) $\Rightarrow$ (3), set $x_1=0$ and $x_k\coloneqq \ruler_{1k}$ for $k>1$.
Then for $i<j$, applying (2) to the triple $(1,i,j)$ gives $\ruler_{1j}=\ruler_{1i}+\ruler_{ij}$, hence
$\ruler_{ij}=x_j-x_i$.
Finally, (3) $\Rightarrow$ (1) follows by substitution.
\end{proof}

Lemma~\ref{lem:triangle-basis} implies that the triangle-consistency constraints in
Proposition~\ref{prop:triangle-ilp} and System~\ref{sys:triangle-lp} may be imposed only for refinements of the
form $(1,j,k)$ with $1<j<k$.
This reduces the number of refinement-indexed constraint blocks from $\binom{n}{3}$ to $\binom{n-1}{2}$ and
the number of triangle variables accordingly.

\paragraph*{Containment-based pruning.}
Define a containment order on intervals by setting $ij\preceq st$ if $s\le i<j\le t$.
Under strict containment $ij\prec st$, every strictly increasing realization satisfies
$x_j-x_i < x_t-x_s$.

Let $m=\binom{n}{2}$ and let $d_1\ge\cdots\ge d_m$ be the multiset $\Delta x$ written as a nonincreasing list
(including multiplicities).
For an interval $ij\in\NxxN$, define $\mathrm{rank}_x(ij)\in[m]$ as any index such that
$d_{\mathrm{rank}_x(ij)}=x_j-x_i$.

\begin{lemma}[Containment-based rank bounds]
\label{lem:containment-rank}
For every strictly increasing $x\in\R^n$ and every $ij\in\NxxN$,
\( i(n-j+1)\le \mathrm{rank}_x(ij)\le m-\binom{j-i+1}{2}+1\).
\end{lemma}

\begin{proof}
There are exactly $i(n-j+1)-1$ strict superintervals of $ij$ (choose $s\in\{1,\dots,i\}$ and
$t\in\{j,\dots,n\}$), each with strictly larger distance, so
$\mathrm{rank}_x(ij)\ge i(n-j+1)$.
There are exactly $\binom{j-i+1}{2}-1$ strict subintervals of $ij$ (intervals strictly contained in
$\{i,\dots,j\}$), each with strictly smaller distance, so at least that many entries appear below $x_j-x_i$ in
the sorted list; hence $\mathrm{rank}_x(ij)\le m-\binom{j-i+1}{2}+1$.
\end{proof}

Write $M_r\coloneqq \sum_{s=1}^r \mu_s$ with $M_0=0$.
In the sorted expansion of $\Y=(y,\mu)$, the value $y_r$ occupies the rank block $[M_{r-1}+1,M_r]$.
Therefore, an interval $ij$ can be assigned label $r$ in any realizable instance only if this block
intersects the admissible rank range from Lemma~\ref{lem:containment-rank}.
Equivalently, we may fix $P^r_{ij}=0$ whenever \(M_r<i(n-j+1)\) or
\(M_{r-1}+1>m-\binom{j-i+1}{2}+1\).
This pruning depends only on $(n,\mu)$ and can be applied as a presolve step in
Proposition~\ref{prop:MILP-turnpike}, Proposition~\ref{prop:triangle-ilp}, and System~\ref{sys:triangle-lp}.

\begin{algorithm}
  \caption{Assignment-first pipeline from unlabeled distances}
  \label{alg:pipeline}
  \begin{algorithmic}[1]
    \REQUIRE Multiset $\Y=(y,\mu)$ with distinct values $y_1>\cdots>y_{m'}>0$.
    \ENSURE Assignment matrix $\widehat P$ (integral if certified) and optional coordinate estimate $\widehat x$.
    \STATE Enumerate $\mathcal{P}_y$ in $O(m'^2)$ time (Section~\ref{subsec:two-partitions}).
    \STATE Solve the triangle \ILP (Proposition~\ref{prop:triangle-ilp}) or its \LP relaxation (System~\ref{sys:triangle-lp}).
    \STATE If an integral assignment $\widehat P$ is obtained, certify realizability and proceed; otherwise treat $\widehat P$ as a relaxation/diagnostic.
    \STATE Optionally extract coordinates $\widehat x$ from the induced ruler (Corollary~\ref{cor:lp-preserves-ruler}) or by downstream regression (Section~\ref{subsec:metrics}).
  \end{algorithmic}
\end{algorithm}

\subsection{Two-partition recovery under noise and rounding}
\label{subsec:two-partition-recovery}

We now return to the bounded-error model from Section~\ref{sec:noisy-variants} to show that, under a simple separation condition, the two-partition structure of a noiseless instance can be recovered from a noisy observation of that instance by a nearest-neighbor rounding procedure.

Let $\Y=(y,\mu)$ be the ground-truth multiset with distinct values $y_1>\cdots>y_{m'}>0$.
For each target index $t\in[m']$, define the \emph{gap at $t$} by
\begin{equation}\label{eq:gap-star}
  \mathrm{gap}_t \coloneqq \min_{\substack{r,s\in[m']\\ y_r+y_s\neq y_t}} |y_r+y_s-y_t|,
  \qquad
  \mathrm{gap}_\star \coloneqq \min_{t\in[m']}\mathrm{gap}_t.
\end{equation}
with the convention $\mathrm{gap}_t=+\infty$ when no invalid pair exists.

We observe noisy distances $\Ytilde$ and form rounded representatives by applying $\operatorname{round}_R$ to each observed value, yielding $\Yhat=(\hat y,\hat\mu)$.
For each $y_k$, the corresponding rounded representative $\hat y_k$ satisfies \(
    \left|\hat y_k-y_k\right| \le r+R.
\)
For use below, the \emph{approximate} two-partition test for threshold $\tau > 0$ is \(
  \left|\hat y_r+\hat y_s-\hat y_t\right| \le \tau.
\)

\begin{theorem}[Two-partition recovery under bounded noise]
  \label{thm:two-partition-recovery}
  Suppose $\mathrm{gap}_\star$---as defined in~\eqref{eq:gap-star}---is such that the noise and rounding radii satisfy $6(r+R) < \mathrm{gap}_\star$.
  With $\tau=\mathrm{gap}_\star/2$, we have for every triple $(r,s,t) \in[m']^{3}$ that \(  
    y_r+y_s=y_t
    \quad\Longleftrightarrow\quad
    |\hat y_r+\hat y_s-\hat y_t|\le \tau.
  \)
  In particular, the exact two-partition set $\mathcal{P}_y$ can be recovered from the rounded noisy representatives~$\hat y$.
\end{theorem}

\begin{proof}
If $y_r+y_s=y_t$, write $\hat y_k=y_k+\varepsilon_k+\delta_k$ with $|\varepsilon_k|\le r$ (measurement error) and $|\delta_k|\le R$ (rounding error). Then $\hat y_r+\hat y_s-\hat y_t=(\varepsilon_r+\varepsilon_s-\varepsilon_t)+(\delta_r+\delta_s-\delta_t)$, which means $|\hat y_r+\hat y_s-\hat y_t|\le 3(r+R)<\tfrac12\mathrm{gap}_\star=\tau$.

If $y_r+y_s\neq y_t$, then $|y_r+y_s-y_t|\ge \mathrm{gap}_\star$ by definition. The same decomposition gives $\hat y_r+\hat y_s-\hat y_t=(y_r+y_s-y_t)+(\varepsilon_r+\varepsilon_s-\varepsilon_t)+(\delta_r+\delta_s-\delta_t)$, and the reverse triangle inequality yields
$|\hat y_r+\hat y_s-\hat y_t|
\ge \mathrm{gap}_\star-3(r+R)
> \tfrac12\mathrm{gap}_\star=\tau$.
\end{proof}

The condition in Theorem~\ref{thm:two-partition-recovery} is deterministic and worst-case: it is phrased in terms of a global separation parameter $\mathrm{gap}_\star$ and a noise bound~$r$.
Its role here is to identify a regime in which the combinatorial input $\mathcal{P}_y$ to the triangle-based \ILP/\LP\ can be recovered exactly from noisy measurements after rounding.
Outside this regime, the recovered two-partition set is not guaranteed to contain the ground truth as a subset; while the formulation remains well-defined in such cases, the output should be interpreted as a triangle-equality--consistent diagnostic/heuristic rather than as a realizability certificate.

\section{Experiments}
\label{sec:experiments}

The experiments below are intended to (i) illustrate empirical integrality behavior of the \LP relaxation on computationally-feasible problem sizes and (ii) demonstrate recovery accuracy as measured by assignment-first metrics enabled by the \ILP/\LP output.

\subsection{Setup}
\label{subsec:exp-setup}

We evaluate the triangle-based \ILP (Proposition~\ref{prop:triangle-ilp}) and the triangle \LP relaxation
(System~\ref{sys:triangle-lp}) on noisy instances under the $(r,R)$ model and small
partial-digest benchmarks subject to bounded-error magnitudes matching those observed in practice~\cite{alizadeh_digestion_protocol_1995}.

All experiments use moderate sizes: synthetic instances use up to $n=30$ (i.e., $m=\binom{30}{2}=435$ distances), and partial-digest benchmarks use up to 100 fragments.

Given an instance $\Y=(y,\mu)$, the \ILP/\LP layer returns an assignment matrix $\Phat$, i.e., a
multi-matching of interval indices $ij$ to distance labels $r$.
Whenever $\widehat P$ is integral, we also compute a coordinate estimate $\widehat x$ by least-squares regression on the induced ruler (Section~\ref{subsec:metrics}), but only as a secondary sanity check since the primary output is the assignment matrix.

Experiments are feasibility problems. 
When using the \LP relaxation, we solve System~\ref{sys:triangle-lp} and record whether the returned solution is integral.
When using the \ILP, we solve Proposition~\ref{prop:triangle-ilp} and record feasibility and recovery metrics.

\paragraph*{Implementation.}
All \LP/\ILP instances were modeled and solved in \texttt{gurobipy} (Gurobi Optimizer 13.0 Python API) using default settings unless otherwise noted.
Experiments were run on \texttt{Apple M4 Pro / 24~GB RAM / macOS Sequoia 15.6.1}.

\subsection{Metrics}
\label{subsec:metrics}

Because the assignment layer returns a discrete labeling, we evaluate recovery primarily using permutation and labeling metrics on $\widehat P$. 
Coordinate errors for $\widehat x$ are treated as secondary.

\paragraph*{Interval labeling error.}
Let $P^\star$ denote a ground-truth assignment for a realizable instance. 
The labeling error is \[
  \mathrm{err}_{\mathrm{lab}}(\widehat P,P^\star)
  \coloneqq \frac{1}{|\NxxN|}\sum_{ij\in\NxxN}\mathbf{1}\!\left[\widehat P_{ij}\neq P^\star_{ij}\right].
\]

\paragraph*{Permutation distance.}
Let $\pi^\star$ be the ground-truth permutation mapping unlabeled distances to interval indices, and let $\widehat\pi$ be the permutation induced by $\Phat$.
We measure the difference between each $\widehat\pi$ and $\pi^\star$ through the normalized Kendall--$\tau$ distance, which we denote as \(
    \mathrm{dist}_{\mathrm{perm}}(\widehat\pi,\pi^\star) \in [0,1].
\)

\paragraph*{Coordinate error.}
For each $\Phat$, we compute the associated $\widehat x$ with least-squares regression as \[
  \widehat x
  \in \arg\min_{x\in\R^n} \sum_{ij\in\NxxN}\Bigl(x_j-x_i-\ruler^{(\widehat P)}_{ij}\Bigr)^2,
\] with translation fixed as $x_1 = 0$.
We then report the mean absolute error (MAE), \[
  \mathrm{MAE}(\widehat x,x^\star)\coloneqq \frac{1}{n}\sum_{i=1}^n \bigl|\widehat x_i-x^\star_i\bigr|,
\] after aligning $\widehat x$ to $x^\star$ up to reflection (i.e., choosing the better of $\widehat x$ and its reflected copy under this metric).

\paragraph*{LP integrality score.}
For the \LP relaxation, we quantify integrality via \[
  \mathrm{int}(\widehat P)
  \coloneqq \frac{1}{|\NxxN|}\sum_{ij\in\NxxN}\max_{r\in[m']} \widehat P^r_{ij},
\] which equals $1$ if and only if $\Phat$ is an integral multi-matching.

\subsection{Synthetic exact instances}
\label{subsec:synthetic-exact}

We generate exact instances by sampling i.i.d.\ coordinates $z_1,\dots,z_n$ from three distributions (uniform on $[0,1]$, standard normal, and Cauchy) and sorting to obtain $x^\star$; for linear instances we form $\Y=\Delta x^\star$, and for circular instances we use the corresponding \Beltway multiset of shortest-arc distances.
Figure~\ref{fig:synthetic-exact} summarizes the integrality score of the triangle \LP relaxation as a function of $n$ for (left) linear \Turnpike instances and (right) circular \Beltway instances.

\begin{figure}[!htb]
  \centering
  \includegraphics[width=0.95\linewidth]{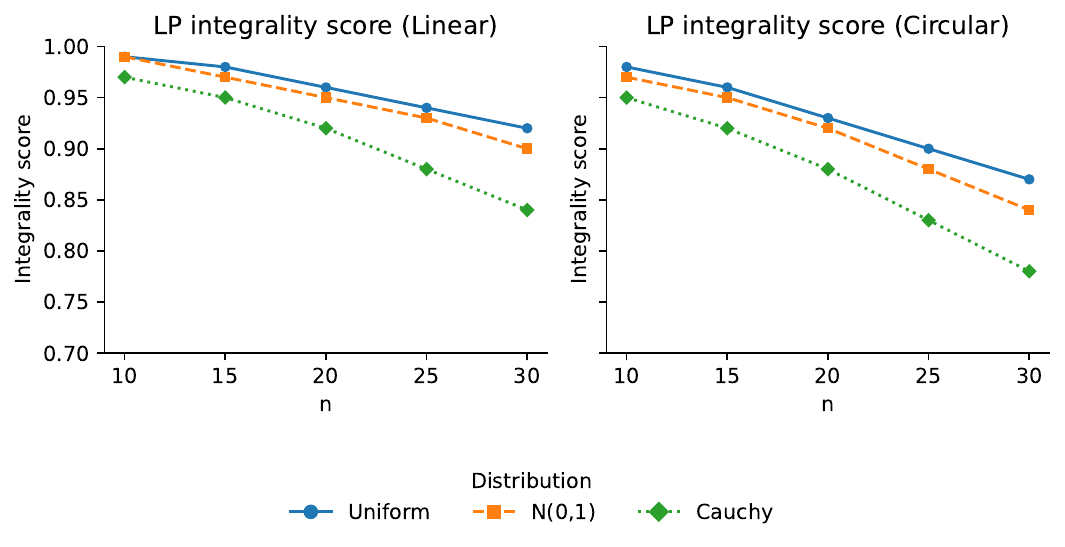}
  \caption{Synthetic exact instances. 
  LP integrality score on (left) linear \Turnpike instances and (right) circular \Beltway instances.
  Curves are stratified by the distribution of $x^\star$.}
  \label{fig:synthetic-exact}
\end{figure}

\subsection{Phase transitions under the \texorpdfstring{$\rR$}{(r,R)} model}
\label{subsec:phase-transition}

We now test the deterministic regime established in Theorem~\ref{thm:two-partition-recovery}.
For a fixed family of exact instances, we add bounded noise of radius $r$ to the distances and then round to a
grid of spacing $R$, yielding $\Yhat=(\hat y,\hat\mu)$ (here hats denote rounded input values; $\Phat$ denotes the solver output).
From $\Yhat$ we enumerate an \emph{approximate} two-partition set via the tolerance test $|\hat y_r+\hat y_s-\hat y_t|\le \tau$ (Theorem~\ref{thm:two-partition-recovery}); we then solve the triangle \ILP/\LP using this approximate $\mathcal{P}_{\hat y}$.

Figure~\ref{fig:noisy-phase} reports three heatmaps over a grid of $(r,R)$ values: the fraction of instances whose exact two-partition set is recovered (top), the rate of spurious two-partitions (middle), and the rate of missing true two-partitions (bottom).
The dashed line indicates the sufficient deterministic condition of Theorem~\ref{thm:two-partition-recovery}.

\begin{figure}[!hb]
  \centering
  \begin{minipage}[t]{0.44\linewidth}
    \centering
    \includegraphics[width=\linewidth]{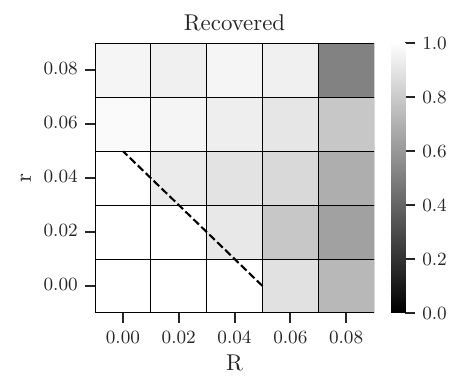}
  \end{minipage}

  \medskip

  \begin{minipage}[t]{0.44\linewidth}
    \centering
    \includegraphics[width=\linewidth]{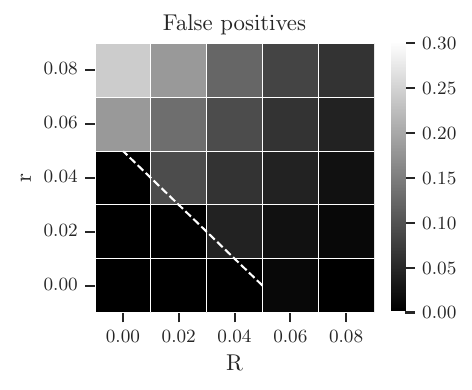}
  \end{minipage}\hfill
  \begin{minipage}[t]{0.44\linewidth}
    \centering
    \includegraphics[width=\linewidth]{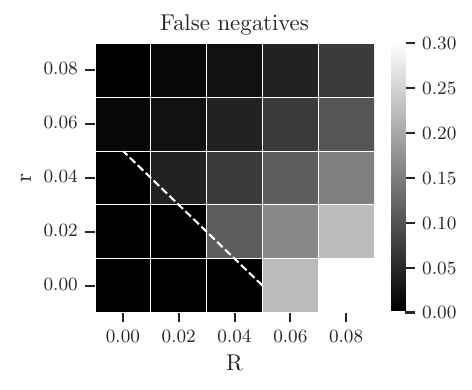}
  \end{minipage}
  \caption{Noisy $(r,R)$ phase diagram for two-partition recovery. Top: recovery rate; bottom-left: false-positive rate (spurious two-partitions); bottom-right: false-negative rate (missing true two-partitions). The dashed line is the sufficient condition of Theorem~\ref{thm:two-partition-recovery}.}
  \label{fig:noisy-phase}
\end{figure}

\FloatBarrier

Takeaway: the recovered two-partition structure is stable in a region consistent with the deterministic bound in Theorem~\ref{thm:two-partition-recovery}.
As $r$ and $R$ increase, $\mathcal{P}_{\hat y}$ accrues false positives and false negatives, so the triangle formulation is fed a distorted two-partition input outside the provable regime.

\subsection{Comparison with regression-based baselines}
\label{subsec:baseline-comparison}

We compare our formulation with regression-based pipelines using baselines established in previous work~\cite{elder_beltway_turnpike_2024}. 
The goal is not to optimize a coordinate loss directly, but to highlight that the triangle \ILP/\LP returns an assignment matrix $\widehat P$ that can be evaluated without committing to a particular numeric norm or optimizer.

Figure~\ref{fig:comparison-baselines} reports median interval labeling error, permutation distance, and coordinate MAE as a function of $n$ for (i)~the triangle-equality \ILP, (ii)~the triangle-equality \LP (rounded to a hard assignment by argmax), (iii)~a majorization--minimization baseline (MM), (iv)~a gradient-descent baseline (GD), and (v)~a sorting-network extended-formulation \ILP baseline (Ext-\ILP). Here Ext-\ILP denotes the implementation used by Elder et al.~\cite{elder_beltway_turnpike_2024}, which is equivalent to the coordinate-coupled \MILP of Proposition~\ref{prop:MILP-turnpike} and is based on Goemans's compact extended formulation for the permutahedron~\cite{goemans_compact_permutahedron_2015}.

\begin{figure}[!htb]
  \centering
\includegraphics[width=0.5\linewidth]{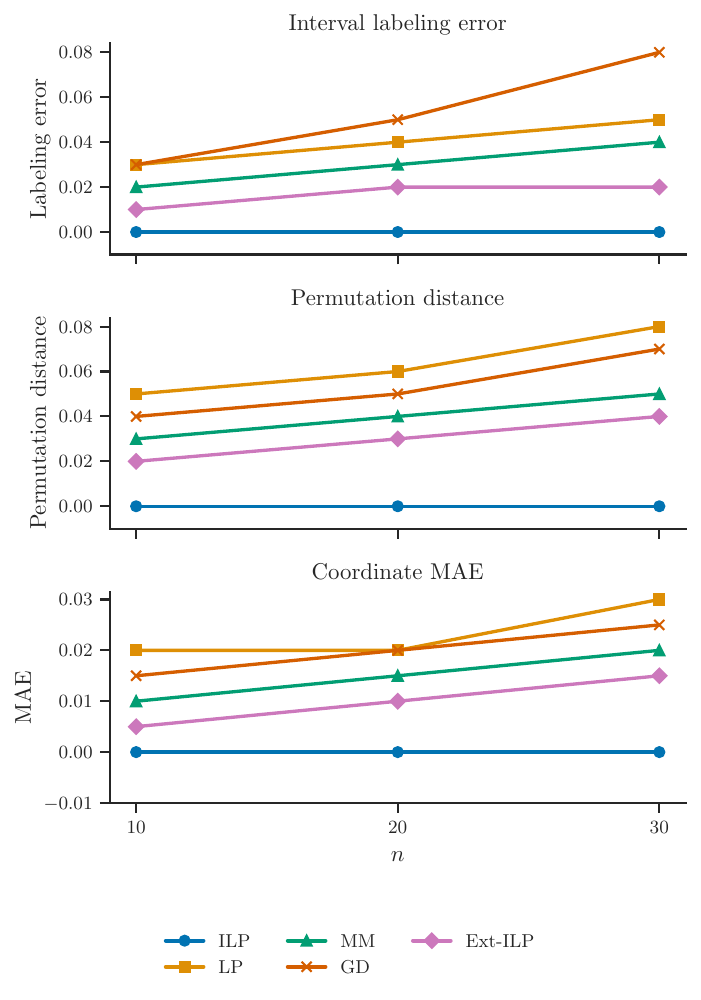}
  \caption{Comparison with regression-based baselines. Panels show median labeling error (top), permutation distance (middle), and coordinate MAE (bottom) for the triangle \ILP/\LP, MM, GD, and the sorting-network Ext-\ILP baseline.}
  \label{fig:comparison-baselines}
\end{figure}

Takeaway: our \ILP formulation directly targets the combinatorial assignment layer, yielding lower labeling and permutation errors across all sizes.

\FloatBarrier

\subsection{Partial digest instances (linear and circular)}
\label{subsec:partial-digest}

We sample partial-digest instances over linear and circular genomes (the
\Beltway variant), observing each up to bounded error with $r = 5$ (i.e., up to 5 additional or missing bases on each fragment), matching error magnitudes observed in practice~\cite{alizadeh_digestion_protocol_1995}.
Figure~\ref{fig:partial-digest} reports normalized labeling error and fragment recovery rate as a function of the number of fragments for linear and circular instances.
Takeaway: for these practical instances, recovery metrics improve as the number of fragments increases while the error radius remains fixed at $r = 5$.

\begin{figure}[!htb]
  \centering
  \includegraphics[width=0.75\linewidth]{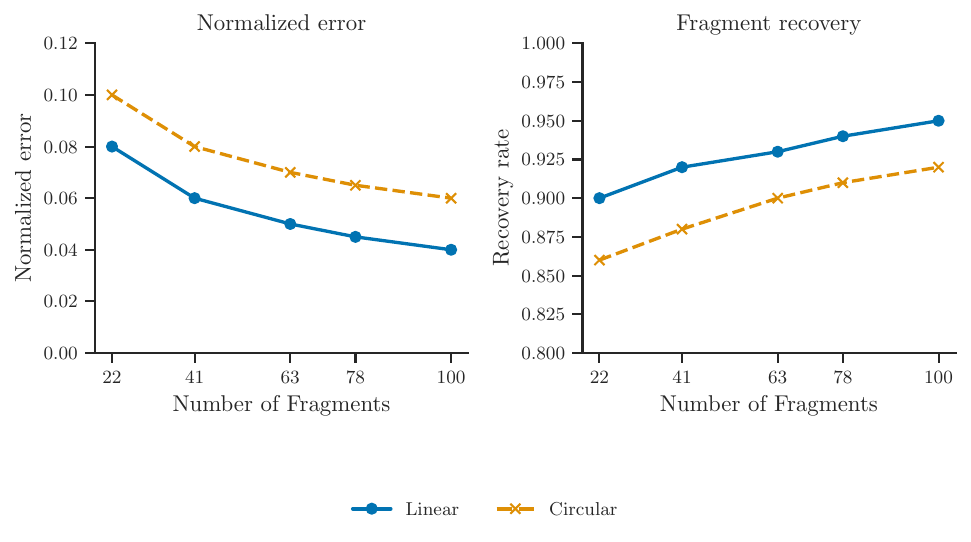}
  \caption{Partial digest experiments on linear and circular genomes. Normalized interval-labeling error (left) and fragment recovery rate (right).}
  \label{fig:partial-digest}
\end{figure}

\FloatBarrier
\section{Conclusion}
\label{sec:conclusion}

Unlabeled-distance reconstruction problems are limited less by arithmetic precision than combinatorial ambiguity: before any regression, one must decide which observed distances correspond to which intervals.
We formulate exact \Turnpike as a triangle-equality assignment problem, yielding a purely combinatorial \ILP and an \LP relaxation parameterized by the two-partition set $\mathcal{P}_y$.
When the relaxation is integral, it provides an explicit certificate and returns an assignment matrix that can be passed to any downstream regression routine.

For uncertain measurements, we prove that bounded noise plus rounding preserve $\mathcal{P}_y$ under a simple separation condition, so the same triangle formulation applies.
Empirically, recovery is strong within this regime and degrades quickly outside it, where the recovered two-partition set can change and the \LP output should be interpreted as a triangle-equality--consistent diagnostic rather than a realizability certificate.

\paragraph*{Future directions}
We see three natural directions: \emph{integrality theory} (conditions guaranteeing integrality of the triangle \LP versus fractional non-integrality); \emph{scalable enforcement} (exploit sparsity in $\mathcal{P}_y$ and separation-based schemes to enforce triangle equalities at larger scales); and \emph{beyond worst-case robustness} (probabilistic analyses, adaptive rounding, and extensions to missing/duplicated distances in applications such as partial digests).

\bibliography{sections/turnpike}

@article{alizadeh_digestion_protocol_1995,
  author  = {Alizadeh, F. and Karp, R. M. and Weisser, D. K. and others},
  title   = {Physical mapping of chromosomes using unique probes},
  journal = {J. Comput. Biol.},
  volume  = {2},
  number  = {2},
  year    = {1995},
  pages   = {159--184},
}

@article{blazewicz_SPDP_hardness_2009,
  author  = {Blazewicz, J. and Burke, E. K. and Kasprzak, M. and Kovalev, A. and Kovalyov, M. Y.},
  title   = {On the approximability of the Simplified Partial Digest Problem},
  journal = {Discrete Appl. Math.},
  volume  = {157},
  number  = {17},
  year    = {2009},
  pages   = {3586--3592},
}

@inproceedings{chen_pairwise_reconstruction_2009,
  author    = {Chen, Shiteng and Huang, Zhiyi and Kannan, Sampath},
  title     = {Reconstructing numbers from pairwise function values},
  booktitle = {ISAAC},
  year      = {2009},
  pages     = {142--152},
}

@article{cieliebak_UPDP_hardness_2005,
  author  = {Cieliebak, M. and Eidenbenz, S. and Penna, P.},
  title   = {Partial digest is hard to solve for erroneous input data},
  journal = {Theor. Comput. Sci.},
  volume  = {349},
  number  = {3},
  year    = {2005},
  pages   = {361--381},
}

@phdthesis{dakic_turnpike_2000,
  author = {Dakic, Tamara},
  title  = {On the Turnpike Problem},
  school = {Simon Fraser University},
  year   = {2000},
}

@article{elder_beltway_turnpike_2024,
  author  = {Elder, C. S. and Hoang, Q. M. and Ferdosi, M. and Kingsford, C.},
  title   = {Approximate and exact optimization algorithms for the Beltway and Turnpike problems with duplicated, missing, partially labeled, and uncertain measurements},
  journal = {J. Comput. Biol.},
  volume  = {31},
  number  = {10},
  year    = {2024},
  pages   = {908--926},
}

@article{fomin_cyclopeptide_circular_sums_2015,
  author  = {Fomin, Eduard},
  title   = {Reconstruction of sequence from its circular partial sums for the cyclopeptide sequencing problem},
  journal = {J. Bioinf. Comput. Biol.},
  volume  = {13},
  number  = {1},
  year    = {2015},
}

@article{gabrys_polymer_ECCs_turnpike_2020,
  author  = {Gabrys, R. and Pattabiraman, S. and Milenkovic, O.},
  title   = {Mass error-correction codes for polymer-based data storage},
  journal = {Proc. ISIT},
  year    = {2020},
  pages   = {25--30},
}

@article{goemans_compact_permutahedron_2015,
  author  = {Goemans, M. X.},
  title   = {Smallest compact formulation for the permutahedron},
  journal = {Math. Program.},
  volume  = {153},
  number  = {1},
  year    = {2015},
  pages   = {5--11},
}

@article{Huang_ReconstructingPointSetsDistributions_2021,
  author  = {Huang, Shuai and Dokmani{\'c}, Ivan},
  title   = {Reconstructing Point Sets from Distance Distributions},
  journal = {IEEE Trans. Signal Process.},
  volume  = {69},
  year    = {2021},
  pages   = {1811--1827},
}

@inproceedings{jaganathan_pairwise_distances_2013,
  author    = {Jaganathan, Kishore and Hassibi, Babak},
  title     = {Reconstruction of integers from pairwise distances},
  booktitle = {ICASSP},
  year      = {2013},
  pages     = {5974--5978},
}

@article{nadimi_fast_PDP_2011,
  author  = {Nadimi, Reza and Fathabadi, Hassan Salehi and Ganjtabesh, Mohammad},
  title   = {A fast algorithm for the partial digest problem},
  journal = {Japan J. Indust. Appl. Math.},
  volume  = {28},
  number  = {2},
  year    = {2011},
  pages   = {315--325},
}

@article{skiena_distances_1990,
  author  = {Skiena, S. S. and Smith, W. D. and Lemke, P.},
  title   = {Reconstructing sets from interpoint distances},
  journal = {Proc. 6th ACM SoCG},
  year    = {1990},
  pages   = {332--339},
}

@article{skiena_PDP_1994,
  author  = {Skiena, S. S. and Sundaram, G.},
  title   = {A partial digest approach to restriction site mapping},
  journal = {Bull. Math. Biol.},
  volume  = {56},
  number  = {2},
  year    = {1994},
  pages   = {275--294},
}

@article{wang_fast_SPDP_2023,
  author  = {Wang, Biing{-}Feng},
  title   = {Fast Algorithms for the Simplified Partial Digest Problem},
  journal = {J. Comput. Biol.},
  volume  = {30},
  number  = {1},
  year    = {2023},
  pages   = {41--51},
}

\end{document}